\documentclass[journal]{IEEEtran}
\usepackage{cite}
\usepackage{amsmath,amssymb,amsthm}
\usepackage{algorithm, algpseudocode}
\usepackage{graphicx}

\usepackage[colorlinks=true]{hyperref}
\usepackage{booktabs} 
\usepackage{array} 
\usepackage{tikz}
\usepackage{comment}
\usetikzlibrary{positioning, arrows.meta}

\newcommand{\pr}[1]{\mathbb{P}\left\{ #1 \right\}}

\newcommand{\mc}[1]{\mathcal{#1}}

\newcommand{\pie}{\pi_{\mathrm{e}}} 

\newcommand\blfootnote[1]{%
  \begingroup
  \renewcommand\thefootnote{}%
  \footnote{#1}%
  \addtocounter{footnote}{-1}%
  \endgroup
}

\newtheorem{thm}{Theorem}
\newtheorem{lem}{Lemma}
\newtheorem{prop}{Proposition}
\newtheorem{cor}{Corollary}

\begin{document}
\title{Harvesting-Aware Routing for Low-Latency Reliable Transmission in Interference-Powered Networks}
\author{Azadeh Sheikholeslami, Allan Nguyen\\
Math and Computer Science Department, Suffolk University, Boston, MA}
\date{}

\maketitle

\begin{abstract}
We consider a multi-hop wireless network where energy-harvesting nodes rely solely on ambient
radio-frequency interference as their power source. The goal is to deliver a packet from source to
destination with minimum latency while guaranteeing a target end-to-end outage probability under
Rayleigh fading. This gives rise to a joint routing and power allocation problem in which interference
plays a dual role: it provides harvestable energy, but also impairs link reliability. We first transform the
intricate, multiplicative end-to-end outage constraint into an additive form. Then, we propose two optimal algorithms. The first, HOPE, minimizes the
bottleneck harvesting latency, i.e., the longest energy-accumulation time along the path, by reducing the
problem to a least-cost path computation on a graph with interference-aware link weights. The second,
THOPE, extends the objective to total end-to-end latency, adding the cumulative transmission delay,
and likewise admits a closed-form least-cost path solution. Both algorithms are solved in polynomial
time via Dijkstra's algorithm and yield explicit expressions for the optimal transmit powers. Extensive
simulations reveal that HOPE and THOPE reduce the average total latency by up to an order of
magnitude compared with shortest-path and minimum-energy benchmarks, while exhibiting sharply
concentrated delay distributions that are largely insensitive to interferer density. 
\end{abstract}

\section{Introduction}
\label{sec:introduction}
\blfootnote{Parts of this work were presented at the 61st Allerton Conference on Communication, Control, and Computing, 2025~\cite{sheikholeslami2025minimum}, and at the 2026 IEEE MIT Undergraduate Research Technology Conference (URTC)~\cite{nguyen2026total}.}
Next-generation wireless networks, spanning industrial IoT, smart cities, disaster-response systems,
and autonomous sensor swarms, must simultaneously satisfy three stringent requirements:
\emph{reliability} (successful data delivery under strict outage limits), \emph{low latency} (timely delivery from
sensing to actuation), and \emph{energy self-sufficiency} (perpetual operation without battery replacements).
Although each of these demands has received substantial attention in isolation, their
joint optimization remains an open challenge, especially when energy is harvested from the
same radio-frequency (RF) environment that also causes interference.

Traditional routing protocols for multi-hop wireless networks assume a continuous power
supply~\cite{zeng2021failure,dehghan2011minimum} and treat interference as a purely detrimental factor~\cite{sheikholeslami2014jamming,sheikholeslami2016energy,bhavathankar2017link,cheng2022revealing}. In energy-harvesting
(EH) networks, however, ambient RF interference can be converted into usable energy~\cite{bi2015wireless,lu2015wireless,zhang2013mimo}.
This paradigm shift---from interference avoidance to interference exploitation---enables
sustainable operation in scenarios where dedicated power infrastructure is impractical, such as
structural health monitoring, precision agriculture, or post-disaster emergency networks. Yet,
harvesting from interference also introduces a fundamental tension: the same interferers that
provide energy also degrade the signal-to-interference-plus-noise ratio (SINR), threatening link
reliability.

Most existing ambient-RF harvesting designs either focus on single-hop links~\cite{guo2016exploiting,ozel2011transmission} or
adopt energy-aware routing heuristics without probabilistic outage guarantees~\cite{lin2007asymptotically,mao2011near}. Few
works simultaneously consider multi-hop routing, dynamic power allocation, and end-to-end
outage constraints. For instance, Thai~\cite{thai2024opportunistic} proposes trust-driven routing that harvests from
jammers, Ghosh~\cite{ghosh2023performance} analyses fixed multi-hop EH topologies, Wang~\cite{wang2024simultaneous} introduces a clustering
protocol for cognitive radio sensors, Shan~\cite{shan2024optimal} develops a harvest-then-transmit schedule, and
Ren~\cite{ren2022opportunistic} presents opportunistic routing with dynamic power. While these works advance the
state of the art, none of them jointly determine both the route and the transmit powers so
as to minimize the end-to-end latency that arises from the \emph{harvesting bottleneck}: the fact that
multi-hop forwarding cannot begin until every node on the path has accumulated sufficient
energy. Moreover, they rarely consider the additional cumulative transmission delay across hops,
which can become significant in dense or long paths.

In this paper, we propose two complementary algorithms that close this gap by jointly optimizing
routing and power allocation under both energy-neutrality and end-to-end outage constraints,
while explicitly treating ambient interference as a harvestable resource.
\begin{itemize}
  \item \emph{HOPE (Harvesting-Optimized Path selection with End-to-end outage guarantees):} minimizes
  the \emph{bottleneck harvesting latency}, i.e., the time required for the slowest node on
  the path to harvest enough energy. We show that the resulting min-max problem can
  be transformed into an equivalent shortest-path problem on a weighted graph, enabling
  polynomial-time solution via Dijkstra's algorithm (Section~\ref{sec:hope}).
  \item \emph{THOPE (Total-latency HOPE):} extends the objective to minimize the \emph{total end-to-end
  latency}, defined as the sum of the bottleneck harvesting time and the cumulative transmission
  time over all hops. Remarkably, this more complex objective also admits a closed-form
  shortest-path formulation, with a modified link weight that accounts for both the harvesting
  time and the constant per-hop transmission duration (Section~\ref{sec:thope}).
\end{itemize}

Both algorithms build on a unified system model in which nodes operate in an energy-neutral
manner, using only the RF energy harvested from ambient interferers during a harvest-then-transmit
cycle. We derive the per-link transmission time from the SINR threshold and packet size, showing
that it becomes network-wide constant under a common rate policy, which allows the entire
problem to be cast as a pure graph-theoretic optimization. Extensive numerical comparisons
against Shortest-Path routing with uniform outage allocation and Minimum-Energy Routing with
outage constraints demonstrate that HOPE and THOPE dramatically reduce total latency (up to
a 10-fold improvement), exhibit sharply concentrated delay distributions, and adapt gracefully
to varying interferer densities, network sizes, and outage requirements (Section~\ref{sec:results}).

The key contributions of this work are as follows. We provide a systematic transformation
of the intricate, multiplicative end-to-end outage constraint into an additive sufficient condition
under Rayleigh fading in Lemma~\ref{lem:e2e_outage}, which decouples link-level variables and enables efficient optimization.
Then, we develop the HOPE framework, which formulates the bottleneck harvesting-latency
minimization as a convex program for a fixed path, yields a closed-form solution via KKT
conditions, and reduces to a shortest-path problem with interference-aware link weights. We extend the framework to total end-to-end latency through THOPE, which additionally accounts
for the cumulative transmission delay, and we derive the corresponding shortest-path weights. 
The link weights for algorithms described above only depend on the aggregate interference power at each node. This property comes from the conservative upper bound used in Lemma~\ref{lem:e2e_outage}. We show that this upper bound is very tight, and using the exact outage probability for routing will not improve the harvesting latency beyond a negligible fraction.
Finally, we present a comprehensive performance evaluation that reveals the gains of
joint routing and power allocation over conventional benchmarks and provides insights into the
trade-offs between harvesting time, hop count, and reliability.

The remainder of the paper is organized as follows. Section~\ref{sec:system} describes the system model,
including the channel, energy-harvesting, and outage characterization. Section~\ref{sec:hope} formulates the
HOPE problem and presents its solution. Section~\ref{sec:thope} extends the framework to THOPE. 
Section~\ref{sec:tightness} analyzes the tightness of the additive outage constraint, and Section~\ref{sec:results}
provides numerical results and comparisons. Finally, Section~\ref{sec:conclusion} concludes the paper and outlines
future research directions.

\section{System Model and Preliminaries}
\label{sec:system}

\subsection{Network Model}
We consider a multi-hop wireless network of energy-harvesting (EH) nodes
deployed at arbitrary positions in the plane, modeled as a directed graph
$\mathcal{G}=(\mathcal{N},\mathcal{E})$, where $\mathcal{N}$ is the set of
nodes and $\mathcal{E}$ is the set of feasible links: a directed link $(u,v)$
belongs to $\mathcal{E}$ if and only if $d_{u,v}\le r_c$, where $d_{u,v}$ is the
Euclidean distance between $u$ and $v$ and $r_c$ is the common communication
range. The nodes have neither
a long-term battery nor an external power supply: their only energy source
is the radio-frequency (RF) interference generated by a set $\mathcal{I}_0$
of interferers at arbitrary positions, which may be intentional jammers or
uncoordinated co-channel transmitters. Interferer $i\in\mathcal{I}_0$
transmits continuously at a constant power $P_i$, and the positions and
powers of the interferers remain fixed over the harvest--transmit cycle of a
packet. We consider the delivery of a single flow; any concurrent traffic is
treated as part of this exogenous interferer set, with a time-averaged power
at each node that is likewise constant over the cycle. All nodes and
interferers use omnidirectional antennas and share a common frequency band.

A packet is to be delivered from a source $S$ to a destination $D$ over a
simple path $\Pi$, that is, a sequence of directed links
$(u,v)\in\mathcal{E}$ in which no node appears twice. We write $u\in\Pi$ for
the transmitting nodes of $\Pi$, namely the source and the relays; each
transmits on exactly one link of the path, and the destination only
receives.

\subsection{Channel Model}
\label{sec:channel}
All channels undergo frequency-flat Rayleigh block fading. The fading
coefficient $h_{u,v}$ between nodes $u$ and $v$ is constant over the
transmission of a packet and independent across links, with $|h_{u,v}|^2$
exponentially distributed and normalized to $\mathbb{E}[|h_{u,v}|^2]=1$; the
coefficients $h_{i,v}$ from interferer $i$ to node $v$ are defined likewise.
When node $u$ transmits, the signal received at node $v$ is
\begin{equation}
y_{v} = \underbrace{h_{u,v}\sqrt{\frac{P_u}{d_{u,v}^\alpha}}\,x_u}_{\text{desired signal}}
+ \underbrace{\sum_{i\in\mathcal{I}_0} h_{i,v}\sqrt{\frac{P_i}{d_{i,v}^\alpha}}\,x_i}_{\text{interference}}
+ n_v,
\label{eq:received_signal}
\end{equation}
where $x_u$ ($x_i$) is the unit-power symbol transmitted by node $u$
(interferer $i$), $P_u$ ($P_i$) is the corresponding transmit power,
$d_{u,v}$ ($d_{i,v}$) is the distance to node $v$, $\alpha$ is the path-loss
exponent, and $n_v\sim\mathcal{CN}(0,N_0)$ is additive white Gaussian noise.
So that the received power cannot exceed the transmitted power in the near
field, every path-loss distance is floored at a reference distance $d_0$; with
a slight abuse of notation, $d_{u,v}$ denotes $\max(d_{u,v},d_0)$ in all
path-loss terms.
Receivers treat interference as noise. Under the single-flow assumption of
Section~\ref{sec:system} and because the nodes of $\Pi$ forward the
packet sequentially (Section~\ref{sec:harvest-then-transmit}), at most one
node transmits at any time, so the only interference in the network is that
generated by $\mathcal{I}_0$.

\subsection{Transmission Rate and Duration}
\label{sec:transmission_time}
All nodes decode with a common SINR threshold $\gamma$
(Section~\ref{sec:outage}) and transmit at the fixed rate
$\rho=\log_2(1+\gamma)$~bits/s/Hz, the rate at which a Gaussian codebook is
decodable whenever the instantaneous SINR is at least $\gamma$. A packet of
$B$ bits sent over a bandwidth $W$ therefore occupies each link for
\begin{equation}
\tau_{u,v} = \frac{B}{W \log_2(1+\gamma)} .
\label{eq:tau_link}
\end{equation}
Because $\gamma$, $B$, and $W$ are common to the whole network, the
transmission time is the same on every link; we denote this constant value
by $\tau$, so that $\tau_{u,v}=\tau$ for all $(u,v)\in\mathcal{E}$.

\subsection{Energy Harvesting and Energy Neutrality}
Each node operates in an \textit{energy-neutral}
manner~\cite{ulukus2015energy}, with only short-term buffering and no
reliance on long-term battery storage\footnote{In practice, we assume each
node is equipped with a low-cost electric double-layer supercapacitor (EDLC)
that can buffer the harvested energy for short durations (on the order of
tens of seconds to a few minutes), sufficient to span the harvest--transmit
cycle. Incorporating the supercapacitor buffer preserves energy neutrality
while enabling brief scheduling flexibility between harvesting and
transmission~\cite{subasinghage2024supercapacitor}.}. During a harvesting interval of duration $T_u$, node $u$ collects the energy
\begin{equation}
E^{\text{harv}}_u = \bar{\Phi}_u\, T_u,
\label{eq:harvested_power}
\end{equation}
where $\bar{\Phi}_u>0$ is the time-averaged power delivered by the harvester of
node $u$. It is a non-decreasing function
\begin{equation}
\bar{\Phi}_u=\Phi(I_u),
\qquad
I_u \triangleq \mathbb{E}\!\left[\sum_{i\in \mathcal{I}_0}
\frac{P_{i}|h_{i,u}|^2}{d_{i,u}^\alpha}\right]
= \sum_{i\in\mathcal{I}_0} \frac{P_i}{d_{i,u}^\alpha},
\label{eq:I_def}
\end{equation}
of the average interference power $I_u$ received at node $u$: the stronger the
interference reaching a node, the more energy it harvests. This is the precise
sense in which interference is a resource
here~\cite{bi2016wireless}, and it is what the routing of
Section~\ref{sec:hope} exploits.

We leave the characteristic $\Phi$ otherwise arbitrary. It may be the response
of any rectifier, including the limited sensitivity and the saturation
exhibited by practical
circuits~\cite{boshkovska2015practical,alevizos2018sensitive,clerckx2019fundamentals},
and every analytical result in this paper holds for any such $\Phi$. For
instance, a practical rectifier may be modeled by $\Phi(I)=\eta I$ between its
sensitivity threshold and its saturation knee, and by a bounded function
beyond it; the results below are unchanged in either regime, since they depend
on $\Phi$ only through the per-node constant $\bar{\Phi}_u$. Only two
properties are used: $\bar{\Phi}_u$ is constant over the harvesting interval,
which holds when that interval spans many fading coherence intervals, and its
value is known to node $u$, which measures it at its own rectifier output. A
node with $\bar{\Phi}_u=0$, for instance one below the sensitivity threshold of
its rectifier, can never fund a transmission and is simply excluded from the
routing graph.

For the numerical evaluation of Section~\ref{sec:results} we adopt the linear
model $\Phi(I_u)=\eta\,I_u$,
where $\eta \in (0,1)$ is the energy conversion efficiency; the model is
accurate between the sensitivity threshold and the saturation level of a
practical rectifier. The parameters of Section~\ref{sec:results} are
normalized rather than drawn from a particular deployment, so the simulations
there compare the schemes across regimes rather than predict absolute
latencies. The contribution of thermal noise to the harvested power is
negligible relative to that of the interferers and is omitted.

Transmitting at power $P_u$ for the duration $\tau$ consumes the energy
$E^{\text{tx}}_u=P_u\tau$. Energy neutrality requires
$E^{\text{tx}}_u\le E^{\text{harv}}_u$, that is,
\begin{equation}
T_u \geq \frac{P_u\tau}{\bar{\Phi}_u}.
\label{eq:energy_neu}
\end{equation}
Within a harvest--transmit cycle, harvested energy may be buffered as
described above; no energy, however, is carried over from one cycle to the
next, so each transmission must be funded entirely by the harvesting interval
that precedes it.
This harvest-use operation is representative of ultra-low-power nodes with
small energy storage~\cite{sudevalayam2011energy} and yields a conservative, worst-case
latency.

We assume each node $v$ can measure the aggregate interference power $I_v$
received at its own antenna, together with the power $\bar{\Phi}_v$ its
harvester delivers, and that these measurements are available to the source
when the route is computed. No knowledge of the number, positions, or powers
of the individual interferers is required.

\subsection{Harvest-then-Transmit Protocol and Latency}
\label{sec:harvest-then-transmit}
Nodes follow a harvest-then-transmit protocol. We assume the route $\Pi$ is
computed before delivery begins and is known to every node on it. At time
$t=0$, every transmitting node of $\Pi$ begins harvesting with no stored
energy, and node $u$ harvests for $T_u$. Once all transmitters have completed harvesting, the
packet is forwarded in a store-and-forward manner, hop by hop from $S$ to
$D$, each hop occupying $\tau$. Forwarding therefore cannot begin until the
slowest transmitter has finished harvesting, as illustrated in
Fig.~\ref{fig:system_model}: the node with the longest harvesting interval
is a bottleneck that dictates the start of forwarding for the entire path.

\begin{figure}
\centering
\includegraphics[width=\columnwidth]{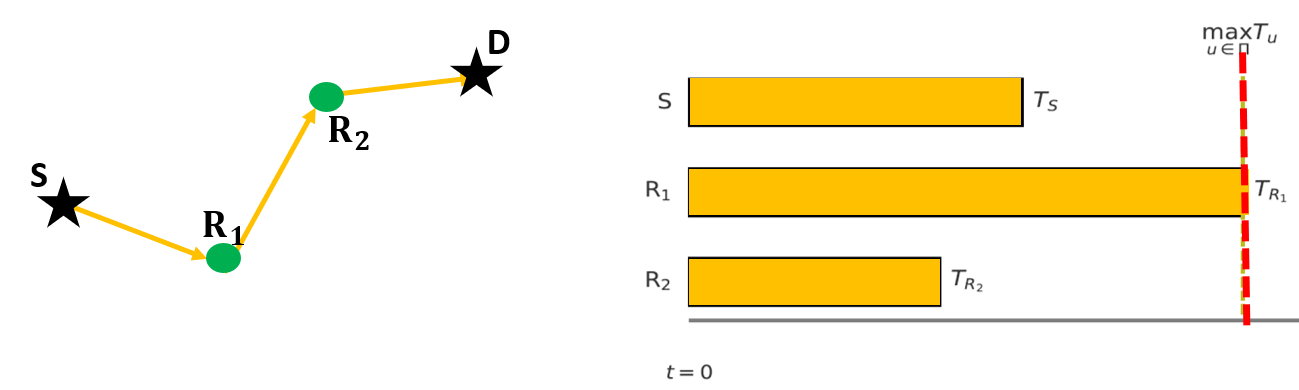}
\caption{Multi-hop energy-harvesting network and bottleneck latency:
The source, $R_1$, and $R_2$ start harvesting at $t=0$; the node with the longest $T_u$ (here $R_1$) determines when end-to-end forwarding can begin.}
\label{fig:system_model}
\end{figure}

\subsection{Outage Probability and End-to-End Reliability}
\label{sec:outage}
Link $(u,v)$ is in outage when the SINR at its receiver falls below the
decoding threshold $\gamma$:
\begin{align}
p_{\text{out}}^{u,v} = \pr{\frac{P_{u}|h_{u,v}|^2/d_{u,v}^\alpha}{N_0+\sum_{i\in \mathcal{I}_0} P_{i}|h_{i,v}|^2/d_{i,v}^\alpha}<\gamma },
\label{eq:p_out_link}
\end{align}
where the sum in the denominator is the instantaneous interference power at
node $v$. Distinct links involve fading coefficients that are independent by
assumption (Section~\ref{sec:channel}), so their outage events are independent, and the end-to-end outage
probability of $\Pi$ is
\begin{equation}
 p_{\text{out}}^{SD} = 1 - \prod_{(u,v)\in\Pi} \left(1 - p_{\text{out}}^{u,v}\right).
 \label{eq:outage_SD}
\end{equation}
We impose the end-to-end reliability constraint
\begin{equation}
p_{\text{out}}^{SD} \leq \pie,
\label{eq:mul_out}
\end{equation}
which, by~\eqref{eq:outage_SD}, reads
\begin{equation}
1 - \prod_{(u,v)\in\Pi} \left(1 - p_{\text{out}}^{u,v}\right) \leq \pie.
\label{eq:outage_product}
\end{equation}

Under Rayleigh fading each factor $1-p_{\text{out}}^{u,v}$ admits a closed
form (Appendix~A) that contains a separate term for
every interferer, and the product in~\eqref{eq:outage_product} couples the
transmit powers of all transmitters on $\Pi$. Taking logarithms turns the
product into a sum over hops, and the resulting constraint is in fact convex
in the transmit powers, since each summand of~\eqref{eq:exact_constraint} is
of the form $a/P_u$ or $\ln(1+b/P_u)$ with $a,b>0$, both convex on $P_u>0$;
the difficulty lies elsewhere. The smallest latency
that~\eqref{eq:outage_product} permits on a given path is defined only
implicitly, as the root of a nonlinear equation, so it does not decompose
into per-link costs and cannot be minimized over paths by a shortest-path
computation; moreover, evaluating it requires the contribution of each
interferer individually. The following lemma removes both obstacles by
replacing~\eqref{eq:outage_product} with a conservative condition that is
linear in each transmitter's inverse power and involves the interferers only
through their aggregate power.

\begin{lem}[Sufficient Condition for End-to-End Outage Constraint]
\label{lem:e2e_outage}
For a given outage threshold $\pie$, let $\epsilon \triangleq -\gamma^{-1}\ln(1-\pie)$.
Then the end-to-end reliability constraint $p_{\text{out}}^{SD} \leq \pie$ is satisfied if
\begin{equation}
\sum_{(u,v)\in\Pi}\frac{d_{u,v}^\alpha}{P_u}(N_0 + I_v) \leq \epsilon,
\label{eq:outage_constraint}
\end{equation}
where $I_v$ is the average interference power at node $v$, defined
in~\eqref{eq:I_def}.
\end{lem}
\begin{proof}
The proof is provided in Appendix~A; it applies
$e^{x}\ge 1+x$ to each interferer's contribution in the exact per-link
expression~\eqref{eq:exact_link} below.
\end{proof}

We refer to~\eqref{eq:outage_constraint} as the \emph{additive} condition, since it is what makes the optimal latency of a path a sum of per-link costs (Section~\ref{sec:hope}).
Because it involves the interferers only through the aggregate $I_v$, the
resulting algorithms require no more than a received-power measurement at
each node, whereas the exact constraint requires each interferer to be
resolved individually. The price of this substitution is quantified in
Section~\ref{sec:tightness}, which shows that it is very nearly tight.

In the next section, we formulate the joint routing and power-control problem
that minimizes the bottleneck harvesting latency while guaranteeing both
energy neutrality and the end-to-end outage constraint.

\section{Minimizing Bottleneck Harvesting Latency (HOPE)}
\label{sec:hope}

Recall from Section~\ref{sec:harvest-then-transmit} that multi-hop transmission from $S$ to $D$ can only commence once \emph{all} nodes on $\Pi$ have harvested sufficient energy. The node requiring the longest harvesting time thus acts as the bottleneck, determining the overall start time of message forwarding. Our goal is to find and characterize the path that minimizes this bottleneck latency while guaranteeing end-to-end reliability.

To this end, we propose HOPE (Harvesting-Optimized Path selection with
End-to-end outage guarantees), which jointly optimizes routing and transmit
power allocation to minimize the \emph{maximum} harvesting time along the
path,
\begin{equation}
\min_{\Pi,{P_u}} \ \max_{u \in \Pi} T_u,
\end{equation}
subject to the energy neutrality constraint~\eqref{eq:energy_neu} and the end-to-end outage constraint~\eqref{eq:outage_constraint}.

For a fixed path $\Pi$, let $\mathcal{T}_{\text{H}}(\Pi)$ denote the minimum of
$\max_{u\in\Pi}T_u$ over transmit powers satisfying the energy-neutrality
condition~\eqref{eq:energy_neu} and the outage
constraint~\eqref{eq:outage_constraint}, that is, the value of the program
\begin{equation}
\begin{aligned}
&\qquad\qquad\min_{P_u} \left( \max_{u \in \Pi} T_u \right), \qquad \text{subject to:} \\
&\quad \sum_{(u,v)\in\Pi}\frac{ d_{u,v}^\alpha}{P_u}(N_0 + I_v)\leq \epsilon,\quad T_u\geq \frac{P_u \tau}{\bar{\Phi}_u} \quad \forall u \in \Pi ,
\label{eq:main_problem_2}
\end{aligned}
\end{equation}
where $\bar{\Phi}_u$ is the power harvested by the transmitter $u$, defined in
Section~\ref{sec:system}. The \emph{routing} problem is then
\begin{equation}
\Pi^\star \in \arg\;\min_{\Pi\in \mc{P}^{SD}}\ \mathcal{T}_{\text{H}}(\Pi),
\end{equation}
where $\mc{P}^{SD}$ is the set of all paths from source $S$ to destination $D$.

This is a mixed-combinatorial and continuous optimization problem. However, due to its structure, it can be solved exactly. The key is to introduce an auxiliary variable $T = \max_{u\in\Pi} T_u$ and apply the KKT conditions. This leads to the following closed-form result.

\begin{thm}[Minimum Bottleneck Harvesting Latency –- HOPE]
\label{thm:hope}
For the HOPE problem defined in~\eqref{eq:main_problem_2} and a given path $\Pi$, the optimal (minimum) bottleneck harvesting latency is
\begin{equation}
\mathcal{T}_{\text{H}}(\Pi) = \frac{\tau}{\epsilon}\;
\sum_{(u,v)\in\Pi}
\frac{d_{u,v}^{\alpha}\,(N_0+I_v)}{\bar{\Phi}_u},
\label{eq:T_HOPE_theorem}
\end{equation}
where $\epsilon = -\ln(1-\pie)/\gamma$. 
Consequently, the globally optimal path $\Pi^\star$ is obtained by solving a shortest-path problem with non-negative link weights $w_{u,v} = \frac{d_{u,v}^{\alpha}(N_0+I_v)}{\bar{\Phi}_u}$.
\end{thm}
\textbf{Proof sketch.}
Introduce an auxiliary variable $T = \max_{u\in\Pi} T_u$. For a fixed path $\Pi$, the problem becomes minimizing $T$ subject to $T \ge T_u$ for all $u \in \Pi$, along with the energy neutrality and the additive outage constraint in~\eqref{eq:outage_constraint}.
The resulting program is convex on $P_u>0$, since each summand
$d_{u,v}^{\alpha}(N_0+I_v)/P_u$ of the outage constraint is convex, and
Slater's condition holds; the KKT conditions are therefore necessary and
sufficient for the global optimum. Solving the stationarity conditions forces
every energy-neutrality constraint tight, so $T=P_u \tau/\bar{\Phi}_u$ for
every node, and the additive outage constraint active. Substituting gives
\[
T = \frac{\tau}{\epsilon}\sum_{(u,v)\in\Pi} \frac{d_{u,v}^\alpha (N_0+I_v)}{\bar{\Phi}_u}.
\]
Minimizing this expression over all paths $\Pi$ from $S$ to $D$ yields the claimed optimal bottleneck latency. $\square$

Theorem~\ref{thm:hope} defines the path cost functional \(\mathcal{T}_{\text{H}}(\Pi)\)  in closed form.
This solution provides a systematic method for joint routing and power allocation, as follows.

\subsection{Routing Algorithm and Power Allocation}
\label{sec:hope_routing}
Let us assign the following weight to each feasible directed link \((u,v)\) in the network:
\begin{equation}
\label{eq:edge_weight}
w_{u,v} \triangleq \frac{d_{u,v}^\alpha (N_0 + I_v)}{\bar{\Phi}_u}\geq 0,
\end{equation}
where nodes with $\bar{\Phi}_u=0$ have been excluded from $\mathcal{G}$ as
described in Section~\ref{sec:system}. Since $N_0>0$ and
$\bar{\Phi}_u>0$, every weight is strictly positive, so the shortest path
returned by Dijkstra's algorithm is simple and no node is revisited.
Because \(\tau/\epsilon\) is a path-independent positive constant, minimizing \(\mathcal{T}_{\text{H}}(\Pi)\) is equivalent to minimizing \(\sum_{(u,v)\in\Pi} w_{u,v}\).
Thus, standard shortest-path algorithms, such as Dijkstra's algorithm, can find a path \(\Pi^\star\) that minimizes \(\sum_{(u,v)\in\Pi} w_{u,v}\) (and therefore \(\mathcal{T}_{\text{H}}(\Pi)\)) in polynomial time, at a cost of \(O(|\mathcal{E}| + |\mathcal{N}|\log|\mathcal{N}|)\) with a binary heap. Note that the number of interferers does not enter this expression: the weight \(w_{u,v}\) depends on the interferers only through the aggregates \(I_u\) and \(I_v\), each of which a node obtains directly by measurement, without resolving the individual interferers or knowing their number.  We return to this point in Section~\ref{sec:tightness}.

Once the minimum-weight path \(\Pi^\star\) is found, the minimized maximum harvesting time is obtained by evaluating~\eqref{eq:T_HOPE_theorem} on \(\Pi^\star\):
\begin{equation}
\mathcal{T}_{\text{H}}(\Pi^\star) = \frac{\tau}{\epsilon}\sum_{(u,v)\in\Pi^\star} w_{u,v}\ .
\end{equation}

For a fixed path, the KKT conditions force all harvesting constraints to be active, i.e. $P_u f_u = T$ with \(f_u=\tau/ \bar{\Phi}_u\).
Hence, the optimal transmit power for each transmitter \(u\in\Pi^\star\) is
\begin{equation}
\label{eq:power_final}
P_u = \frac{\mathcal{T}_{\text{H}}(\Pi^\star)}{f_u}
= \frac{\bar{\Phi}_u}{\tau}\,\mathcal{T}_{\text{H}}(\Pi^\star).
\end{equation}
Substituting \(\mathcal{T}_{\text{H}}(\Pi^\star)\) yields 
\begin{equation}
P_u = \frac{\bar{\Phi}_u}{\epsilon}\,
\sum_{(u',v')\in\Pi^\star}\frac{d_{u',v'}^{\alpha}\,(N_0+I_{v'})}{\bar{\Phi}_{u'}}.
\end{equation}

This procedure yields a computationally efficient joint routing and power-allocation solution 
for minimum harvesting bottleneck latency,
as summarized in Algorithm~\ref{alg:hope_short}.

\begin{algorithm}
\small
\caption{HOPE: Minimum-Latency Reliable Routing}
\label{alg:hope_short}
\begin{algorithmic}[1]
\State \textbf{Precompute} average interference $I_v = \sum_{i \in \mathcal{I}_0} P_i / d_{i,v}^\alpha$ for all nodes $v$.

\State \textbf{Set} $\epsilon = -\gamma^{-1} \ln(1 - \pie)$.

\State \textbf{Assign weight} to each link $(u,v)$: $w_{u,v} = \dfrac{d_{u,v}^\alpha (N_0 + I_v)}{\bar{\Phi}_u}$.

\State \textbf{Find path} $\Pi^\star = \operatorname{Dijkstra}(\mathcal{G}, S, D, w)$.

\State \textbf{Compute} max harvesting time: $\mathcal{T}_{\text{H}}(\Pi^\star) = \dfrac{\tau}{\epsilon} \sum_{(u,v) \in \Pi^\star} w_{u,v}$.

\For{each transmitter $u$ on $\Pi^\star$}

    \State $P_u = \dfrac{\bar{\Phi}_u}{\tau} \mathcal{T}_{\text{H}}(\Pi^\star)$
    
\EndFor

\State \Return $\Pi^\star$, $\mathcal{T}_{\text{H}}(\Pi^\star)$, $\{P_u\}$

\end{algorithmic}
\end{algorithm}

\textbf{Discussion.}
HOPE minimizes only the time at which forwarding can begin. The packet must
then traverse every hop of $\Pi^\star$, and a route that exploits strong
harvesting regions may contain many short hops, so that, depending on the
regime, the cumulative transmission delay $|\Pi^\star|\tau$ need not be
negligible relative to the harvesting latency. Applications with hard
deadlines, such as industrial monitoring or emergency response, must account
for both components, which motivates the total-latency objective of the next
section.

\section{The Extended Problem: Minimizing Total Latency (THOPE)}
\label{sec:thope}

In this section, we propose THOPE (Total‑latency HOPE), which extends the objective to the total end‑to‑end latency. 
THOPE minimizes the total time from the start of harvesting until the destination receives the first packet. 
This total latency of a fixed path $\Pi$ is the sum of the initial bottleneck harvesting latency and the cumulative transmission delay over all hops of $\Pi$, i.e.
\[
\text{Total latency} \ =\ \max_{u \in \Pi} T_u + \sum_{(u,v) \in \Pi} \tau_{u,v}.
\]
We define the total latency cost functional as the smallest achievable total latency along $\Pi$:
\[
\mathcal{T}_{\text{tot}}(\Pi) \;\triangleq\; \min_{P_u} \Big( \max_{u \in \Pi} T_u + \sum_{(u,v) \in \Pi} \tau_{u,v}\Big).
\]

As defined in Section~\ref{sec:transmission_time}, all links use the same transmission time $\tau$.
By incorporating the energy neutrality constraint~\eqref{eq:energy_neu} and the end-to-end outage constraint~\eqref{eq:outage_constraint}, for a given path $\Pi$, the optimization problem is formulated as:
\begin{equation}
\begin{aligned}
&\qquad\qquad\min_{P_u} \Big( \max_{u \in \Pi} T_u + \sum_{(u,v) \in \Pi} \tau\Big), \quad\text{subject to:} \\
&\quad{\sum_{(u,v)\in\Pi}\frac{ d_{u,v}^\alpha}{P_u}(N_0 + I_v)}\leq \epsilon,
 \quad T_u\geq \frac{P_u \tau}{\bar{\Phi}_u} \quad \forall u \in \Pi
\label{eq:total_problem_1}
\end{aligned}
\end{equation}

Following a similar KKT-based approach, we derive the optimal solution for this more complex, non-separable objective.

\begin{thm}[Minimum Total Latency –- THOPE]
\label{thm:thope}
For the THOPE problem defined in~\eqref{eq:total_problem_1}, the optimal (minimum) total end‑to‑end latency is
\begin{equation}
\label{eq:T_THOPE_theorem}
\mathcal{T}_{\text{tot}}(\Pi) \;=\;
\sum_{(u,v)\in\Pi} \tau\left(
\frac{d_{u,v}^{\alpha}\,(N_0+I_v)}{\epsilon\,\bar{\Phi}_u} + 1
\right),
\end{equation}

where $\epsilon = -\ln(1-\pie)/\gamma$,
and all other symbols are as defined in Section~\ref{sec:system}.
\end{thm}

\textbf{Proof sketch.} Introduce $T = \max_{u\in\Pi} T_u$. For a fixed path $\Pi$, the objective becomes $T + |\Pi|\times \tau$, where $|\Pi|$ is the number of hops in the path. Using Lagrangian duality and KKT conditions, optimality forces 
$P_u = T\bar{\Phi}_u/\tau$
for every node and the outage constraint to be active; the objective differs
from that of Theorem~\ref{thm:hope} by the constant $|\Pi|\tau$, so convexity
and Slater's condition carry over and the KKT conditions again identify the
global optimum. Substituting yields 
$$T = \frac{\tau}{\epsilon}\sum_{(u,v)\in\Pi}\frac{d_{u,v}^\alpha(N_0+I_v)}{\bar{\Phi}_u}.$$
Hence, the total latency is 
$$T + |\Pi|\times \tau = \sum_{(u,v)\in\Pi}\tau\Bigl(\frac{d_{u,v}^\alpha(N_0+I_v)}{\epsilon\,\bar{\Phi}_u}+1\Bigr).$$
Minimizing over all paths gives the claimed optimal total latency. $\square$

\subsection{Routing Algorithm and Power Allocation}

The expression in~\eqref{eq:T_THOPE_theorem} is a sum of non-negative edge
weights, so minimizing $\mathcal{T}_{\text{tot}}(\Pi)$ over all paths from $S$
to $D$ is equivalent to finding a shortest path in a graph where each directed
link $(u,v)$ carries the weight
\[
w_{u,v} \;\triangleq\; \tau\left(\frac{d_{u,v}^\alpha (N_0 + I_v)}{\epsilon\,\bar{\Phi}_u} + 1\right) \;\ge\; 0.
\]
The first term is the harvesting cost of the hop, inherited from HOPE's
weight~\eqref{eq:edge_weight} scaled by $\tau/\epsilon$; the second is the
per-hop transmission time $\tau$. THOPE's route therefore minimizes the sum of
harvesting and transmission costs, and a route that is short in hops is
favored over one that is long, even when the latter offers better harvesting.
Because all weights are non‑negative, a standard shortest‑path algorithm such as Dijkstra's algorithm can compute the optimal path $\Pi^\star$ in polynomial time (e.g., $O(|\mathcal{E}| + |\mathcal{N}|\log|\mathcal{N}|)$ with a binary heap).

THOPE is therefore Algorithm~\ref{alg:hope_short} with three changes: line~3
uses the weight $w_{u,v}$ above, line~5 computes
$\mathcal{T}_{\text{tot}}(\Pi^\star)=\sum_{(u,v)\in\Pi^\star}w_{u,v}$, and
line~7 uses $T=\mathcal{T}_{\text{tot}}(\Pi^\star)-|\Pi^\star|\tau$ in place of
$\mathcal{T}_{\text{H}}(\Pi^\star)$.

Once the optimal path $\Pi^\star$ is found, using the same reasoning as in Section~\ref{sec:hope_routing}, the corresponding optimal transmit power for each transmitter $u$ on $\Pi^\star$ is
\[
P_u \;=\; \frac{T}{f_u}
      \;=\; \frac{\bar{\Phi}_u}{\tau}\; T.
\]
Substituting the closed‑form expression for $T$ gives
\[
P_u \;=\; \bar{\Phi}_u\;
        \sum_{(u',v')\in\Pi^\star}
          \frac{d_{u',v'}^{\alpha}\,(N_0+I_{v'})}
               {\epsilon\,\bar{\Phi}_{u'} } .
\]
These powers satisfy the energy‑neutrality condition $T = P_u f_u = \mathcal{T}_{\text{tot}}(\Pi^\star) - \sum_{(u,v)\in\Pi^\star}\tau$ (i.e., the harvesting latency is the same for all nodes on the path, and the total latency equals the sum of that common harvesting latency plus the total transmission time).  Moreover, these powers satisfy the additive condition~\eqref{eq:outage_constraint},
so by Lemma~\ref{lem:e2e_outage} the end-to-end outage probability is at
most $\pie$; since that condition is conservative, the true outage
probability lies below the target, and Section~\ref{sec:tightness} shows
how little latency this conservatism costs.

\section{Tightness of the Additive Outage Constraint}
\label{sec:tightness}

Lemma~\ref{lem:e2e_outage} replaces the exact, multiplicative outage constraint
by a conservative additive sufficient condition, and this substitution is what
makes the closed-form solutions of Theorems~\ref{thm:hope} and~\ref{thm:thope}
possible. Because both algorithms operate their nodes at the boundary of that
condition, any looseness in Lemma~\ref{lem:e2e_outage} translates directly into
transmit power that is larger than necessary, and hence into harvesting latency
that is larger than necessary. In this section we quantify that looseness. 
\subsection{Exact End-to-End Constraint}
Throughout this section we write
\begin{equation}
\Lambda \;\triangleq\; \gamma\,\epsilon \;=\; -\ln(1-\pie),
\qquad
I_{i,v} \;\triangleq\; \frac{P_i}{d_{i,v}^{\alpha}},
\label{eq:tight_defs}
\end{equation}
so that $I_{i,v}$ is the average power received at node $v$ from the single
interferer $i$ and $I_v=\sum_{i\in\mc{I}_0}I_{i,v}$.
We define the
per-interferer interference-to-signal ratio on link $(u,v)$,
\begin{equation}
x_{i,v} \;\triangleq\; \frac{\gamma\,d_{u,v}^{\alpha}\,I_{i,v}}{P_u}.
\label{eq:x_def}
\end{equation}

The derivation in Appendix~A establishes, before any
bounding is applied, the exact per-link success probability under Rayleigh
fading,
\begin{equation}
1-p_{\text{out}}^{u,v}
=\frac{1}{\prod_{i\in\mc{I}_0}\bigl(1+x_{i,v}\bigr)}
\exp\Bigl(-\frac{\gamma\,d_{u,v}^{\alpha}N_0}{P_u}\Bigr).
\label{eq:exact_link}
\end{equation}
Substituting~\eqref{eq:exact_link} into the end-to-end
definition~\eqref{eq:outage_SD} and taking logarithms, the constraint
$p_{\text{out}}^{SD}\le\pie$ is equivalent to
\begin{equation}
\Psi \;\triangleq\;
\sum_{(u,v)\in\Pi}\Bigl[\,\frac{\gamma\,d_{u,v}^{\alpha}N_0}{P_u}
+\sum_{i\in\mc{I}_0}\ln\bigl(1+x_{i,v}\bigr)\Bigr]\;\le\;\Lambda .
\label{eq:exact_constraint}
\end{equation}
The gap between~\eqref{eq:exact_constraint} and the additive condition
of Lemma~\ref{lem:e2e_outage} comes from the per-interferer relaxation $\ln(1+x)\le x$.

Both HOPE and THOPE drive the energy-neutrality constraint to equality, i.e.
$P_u = \bar{\Phi}_u T/\tau$ with $T=\max_{u\in\Pi}T_u$
(Theorems~\ref{thm:hope} and~\ref{thm:thope}). Substituting this
into~\eqref{eq:x_def} gives $x_{i,v}=c_{u,v}I_{i,v}/T$ with the link constant
\begin{equation}
c_{u,v}\;\triangleq\;\frac{\gamma\,\tau\,d_{u,v}^{\alpha}}{\bar{\Phi}_u}.
\label{eq:c_def}
\end{equation}
Hence, in terms of the reciprocal bottleneck latency $z\triangleq 1/T$, the
exact constraint~\eqref{eq:exact_constraint} is
\begin{equation}
\Psi(z)=\sum_{(u,v)\in\Pi}\Bigl[c_{u,v}N_0\,z
+\sum_{i\in\mc{I}_0}\ln\bigl(1+c_{u,v}I_{i,v}z\bigr)\Bigr]\;\le\;\Lambda .
\label{eq:psi_z}
\end{equation}
Setting $P_u = \bar{\Phi}_u T/\tau$ and substituting $z$ and $c_{u,v}$ in the right hand side of additive condition of~\eqref{eq:outage_constraint} in
Lemma~\ref{lem:e2e_outage} gives 
\begin{equation}
B(z)\;\triangleq\;z\sum_{(u,v)\in\Pi}c_{u,v}\bigl(N_0+I_v\bigr)
\;=\;\Lambda\,\mathcal{T}_{\text{H}}(\Pi)\,z,
\label{eq:B_def}
\end{equation}

Hence, the comparison in this section is between the linear function $B$ and the exact function~$\Psi$.

Let us denote the bottleneck harvesting latency of path $\Pi$ under the exact
constraint~\eqref{eq:psi_z} by
\begin{equation}
\mathcal{T}_{\text{H}}^{\text{exact}}(\Pi)
\;\triangleq\;\min\bigl\{\,T>0 \;:\; \Psi(1/T)\le\Lambda \,\bigr\}
\label{eq:Tex_def}
\end{equation}
which compared to $\mathcal{T}_{\text{H}}(\Pi)$ of
Theorem~\ref{thm:hope}, is the same quantity obtained from the exact outage constraint. 

\subsection{Structural Properties}

Fig.~\ref{fig:tangent} summarizes the geometry that underlies the remainder of
this section. Both formulations $B$ and $\Psi$ spend the same outage budget $\Lambda$, but they
measure the expenditure differently: $B$ overstates the true exponent $\Psi$, and thus reaches the budget early, at $z_0=1/\mathcal{T}_{\text{H}}(\Pi)$, and HOPE stops
there.
On the other hand, $\Psi$ does not reach $\Lambda$ until
$z_{\text{exact}}=1/\mathcal{T}_{\text{H}}^{\text{exact}}(\Pi)$. 
The horizontal separation between $z_0$ and $z_{\text{exact}}$ is the recoverable latency.
This value is
governed entirely by the vertical difference $B-\Psi$. The following proposition makes this precise.

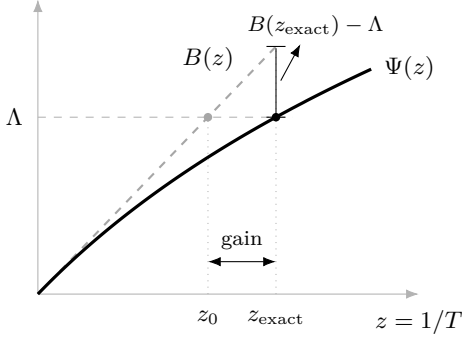
\begin{figure}[t]
\centering
\begin{tikzpicture}[x=0.90cm, y=0.78cm, >=Latex]
  \draw[->, gray!60] (0,0) -- (5.6,0);
  \draw[->, gray!60] (0,0) -- (0,5.0);
  \draw[gray!60, dashed] (0,3) -- (3.75,3);
  \draw[gray!70, dashed, thick] (0,0) -- (3.5,4.2);
  \draw[black, very thick, samples=120, smooth, domain=0:4.9]
        plot (\x, {4.706*ln(1+0.255*\x)});
  \draw[gray!60, dotted] (2.5,0) -- (2.5,3);
  \draw[gray!60, dotted] (3.5,0) -- (3.5,4.2);
  \fill[gray!70] (2.5,3) circle (1.6pt);
  \fill[black]   (3.5,3) circle (1.6pt);
  \draw[black, thin] (3.36,4.2) -- (3.64,4.2);
  \draw[black, thin] (3.36,3.0) -- (3.64,3.0);
  \draw[black, thin]  (3.5,3.03) -- (3.5,4.17);
  \draw[->, thin] (3.57,3.62) -- (3.86,4.24);
  \node[font=\footnotesize] at (4.08,4.55) {$B(z_{\text{exact}})-\Lambda$};
  \draw[<->] (2.5,0.55) -- (3.5,0.55);
  \node[above, font=\footnotesize] at (3.0,0.60) {gain};
  \node[left,  font=\small] at (-0.08,3) {$\Lambda$};
  \node[below, font=\small] at (2.5,-0.12) {$z_0$};
  \node[below, font=\small] at (3.5,-0.12) {$z_{\text{exact}}$};
  \node[below, font=\small] at (5.6,-0.12) {$z=1/T$};
  \node[above left,  font=\small] at (3.02,3.62) {$B(z)$};
  \node[right, font=\small] at (4.95,3.83) {$\Psi(z)$};
\end{tikzpicture}
\caption{The exact outage exponent $\Psi$ of~\eqref{eq:psi_z} (solid) and its
linear counterpart $B$ of~\eqref{eq:B_def} (dashed), as functions of the
reciprocal bottleneck latency $z=1/T$. $B$ is the tangent to $\Psi$ at the
origin, so it lies above $\Psi$ and meets the budget $\Lambda$ first, at
$z_0=1/\mathcal{T}_{\text{H}}(\Pi)$; the exact constraint is met only at
$z_{\text{exact}}=1/\mathcal{T}_{\text{H}}^{\text{exact}}(\Pi)$. The slack
$B(z_{\text{exact}})-\Lambda$ is the amount discarded by the relaxation
$\ln(1+x)\le x$; Theorem~\ref{thm:gap} bounds it, and hence the horizontal
gain, by the fraction $\bar{x}/2$ of $B(z_{\text{exact}})$.}
\label{fig:tangent}
\end{figure}

\begin{prop}[Concavity, tangency, and monotone refinement]
\label{prop:concave}
For any path $\Pi$, the function $\Psi$ in~\eqref{eq:psi_z} satisfies:
\begin{enumerate}
\item[(i)] $\Psi(0)=0$, and $\Psi$ is strictly increasing and, provided
$\mc{I}_0\neq\emptyset$, strictly concave on $z\ge0$. Consequently,
$\Psi(z)=\Lambda$ has a unique root $z_{\text{exact}}$ and
$\mathcal{T}_{\text{H}}^{\text{exact}}(\Pi)=1/z_{\text{exact}}$. In the
degenerate case $\mc{I}_0=\emptyset$ the logarithmic terms are absent, $\Psi$
coincides with $B$, and the gap studied in this section vanishes;
\item[(ii)] $B$ is the tangent to $\Psi$ at the origin. Hence,
$B(z)\ge\Psi(z)$ for all $z\ge0$;
\item[(iii)] starting from $z_0=1/\mathcal{T}_{\text{H}}(\Pi)$, the Newton
iteration
\begin{align}
z_{k+1}&=z_k+\frac{\Lambda-\Psi(z_k)}{\Psi'(z_k)},
\label{eq:newton}\\
\Psi'(z)&=\sum_{(u,v)\in\Pi}\Bigl[c_{u,v}N_0+\sum_{i\in\mc{I}_0}
\frac{c_{u,v}I_{i,v}}{1+c_{u,v}I_{i,v}z}\Bigr],
\label{eq:psi_prime}
\end{align}
satisfies $z_k\le z_{k+1}\le z_{\text{exact}}$ for every $k$.
\end{enumerate}
\end{prop}

\begin{proof}
The proof is provided in Appendix~B.
\end{proof}

Part (iii) has a practical consequence that matters more than the convergence
rate: because every iterate satisfies $z_k\le z_{\text{exact}}$, every iterate is
\emph{feasible}. Therefore, the refinement may be truncated after any number of
steps, without ever violating the end-to-end outage
guarantee.

\subsection{A Bound on the Achievable Latency Gap}

The refinement of Proposition~\ref{prop:concave} part (iii) converges to
$\mathcal{T}_{\text{H}}^{\text{exact}}(\Pi)$, which is limited by the target outage constraint.
The following theorem quantifies the latency gap between $\mathcal{T}_{\text{H}}^{\text{exact}}(\Pi)$ and $\mathcal{T}_{\text{H}}(\Pi)$.

\begin{thm}[Latency gap of the additive bound]
\label{thm:gap}
Let $\pie<2/3$ be the target end-to-end outage probability
of~\eqref{eq:mul_out}, and let
\begin{equation}
\bar{x}\;\triangleq\;e^{\Lambda}-1\;=\;\frac{\pie}{1-\pie}
\label{eq:xbar_def}
\end{equation}
denote the largest interference-to-signal ratio $x_{i,v}$ that any single
interferer can present at any receiver on a path without, by itself, exhausting
the outage budget. Then, for \emph{every} path $\Pi$,
\begin{equation}
\mathcal{T}_{\text{H}}^{\text{exact}}(\Pi)\;\ge\;
\Bigl(1-\frac{\bar{x}}{2}\Bigr)\,\mathcal{T}_{\text{H}}(\Pi).
\label{eq:gap_bound}
\end{equation}
\end{thm}

\textbf{Proof sketch.} Because all summands of $\Psi$ are non-negative and sum
to $\Lambda$ at the optimum, no single term may exceed $\Lambda$, which bounds
every $x_{i,v}$ by $\bar{x}$. Applying $x-\ln(1+x)\le x^2/2$ term by term and
using that bound converts the discarded slack into a fraction $\bar{x}/2$ of
$B(z_{\text{exact}})$, and~\eqref{eq:B_def} turns the resulting inequality into
the claimed bound on the latency. The detailed proof is given in
Appendix~C. $\square$

In words, $\bar{x}$ is the largest ratio at which a single interferer can be present without causing outage on its own, and Theorem~\ref{thm:gap} says that the additive condition is loose only to the extent that some interferer approaches this limit.

Theorem~\ref{thm:gap} compares the two latencies on a common path. 
However, because it
holds for \emph{every} path, it also bounds what could be gained by re-optimizing the path under the exact constraint.

\begin{cor}[No gain from re-routing]
\label{cor:gap_global}
Under the hypotheses of Theorem~\ref{thm:gap}, let $\Pi^{\text{exact}}$ minimize
$\mathcal{T}_{\text{H}}^{\text{exact}}$ over all paths from $S$ to $D$, that is,
let it be the best route available to a scheme that uses the exact constraint
for both power allocation and routing. Then
\begin{equation}
\mathcal{T}_{\text{H}}^{\text{exact}}(\Pi^{\text{exact}})
\;\ge\;\Bigl(1-\frac{\bar{x}}{2}\Bigr)\,\mathcal{T}_{\text{H}}(\Pi^\star).
\label{eq:gap_global}
\end{equation}
No scheme within the model of Section~\ref{sec:system}, however it
allocates power and however it routes, can therefore reduce HOPE's bottleneck
latency by more than the fraction $\bar{x}/2$ of its value.
\end{cor}

\begin{proof}
Theorem~\ref{thm:gap} holds for every path, in particular for
$\Pi^{\text{exact}}$, giving $\mathcal{T}_{\text{H}}^{\text{exact}}
(\Pi^{\text{exact}})\ge(1-\bar{x}/2)\,\mathcal{T}_{\text{H}}
(\Pi^{\text{exact}})$. The right-hand side is a latency under the additive condition, and
$\Pi^\star$ minimizes $\mathcal{T}_{\text{H}}$ over all paths by
Theorem~\ref{thm:hope}, so $\mathcal{T}_{\text{H}}(\Pi^{\text{exact}})\ge
\mathcal{T}_{\text{H}}(\Pi^\star)$. Combining the two
gives~\eqref{eq:gap_global}.
\end{proof}

Theorem~\ref{thm:gap} is universal but coarse. The following
proposition complements it with a bound that is typically sharper by orders of magnitude (as will be shown in Section~\ref{sec:results}).

\begin{prop}[Computable per-instance bound]
\label{prop:instbound}
Fix a path $\Pi$, let $z_0=1/\mathcal{T}_{\text{H}}(\Pi)$, and define
\begin{equation}
G(\Pi)\;\triangleq\;\frac{1}{\Lambda}
\sum_{(u,v)\in\Pi}\sum_{i\in\mc{I}_0}
\bigl(c_{u,v}I_{i,v}z_0\bigr)^{2},
\label{eq:G_def}
\end{equation}
that is, the sum of the squared ratios $x_{i,v}$ evaluated at $z_0$, normalized by the outage budget. If $\pie\le1/2$ and $G(\Pi)\le1/2$,
then
\begin{equation}
1-\frac{\mathcal{T}_{\text{H}}^{\text{exact}}(\Pi)}{\mathcal{T}_{\text{H}}(\Pi)}
\;\le\;\frac{1-\sqrt{1-2G(\Pi)}}{2}
\;=\;\frac{G(\Pi)}{2}+O\bigl(G(\Pi)^2\bigr).
\label{eq:instbound}
\end{equation}
\end{prop}

\begin{proof}
The proof is provided in Appendix~D.
\end{proof}

The two bounds are of different kinds, and each is bought at a price.
The bound of Theorem~\ref{thm:gap} and Corollary~\ref{cor:gap_global} is pessimistic, but it requires no knowledge
of the network. They only depend on the outage target $\pie$  and hold for every network.
The bound of Proposition~\ref{prop:instbound} is much sharper, but it needs the output of
HOPE or THOPE and the placement of the interferers.

\subsection{Refinement Procedure}

Proposition~\ref{prop:concave} part (iii) yields a refinement that is common to both
proposed algorithms. Both HOPE and THOPE terminate with the same two
operations: a bottleneck latency $T$ computed from the outage budget on the
selected path, and the transmit powers $P_u=\bar{\Phi}_u T/\tau$ derived from it.
The refinement acts only on the first of these, leaving the routing decision and
the power-allocation rule untouched, and is summarized in
Algorithm~\ref{alg:refine}.

\begin{algorithm}
\small
\caption{Exact refinement of the bottleneck latency (common to HOPE and THOPE)}
\label{alg:refine}
\begin{algorithmic}[1]
\Require path $\Pi^\star$ and bottleneck latency $T$, from
Algorithm~\ref{alg:hope_short} (line~5) for HOPE, or from THOPE net of the
transmission time $|\Pi^\star|\tau$
\State \textbf{Set} $\Lambda=-\ln(1-\pie)$, \; $z \gets 1/T$
\State \textbf{Precompute} $c_{u,v}=\gamma\tau d_{u,v}^{\alpha}/(\bar{\Phi}_u)$ and
$I_{i,v}=P_i/d_{i,v}^{\alpha}$ for all $(u,v)\in\Pi^\star$, $i\in\mc{I}_0$
\Repeat
  \State $\Psi \gets \sum_{(u,v)\in\Pi^\star}\bigl[c_{u,v}N_0 z
         +\sum_{i}\ln(1+c_{u,v}I_{i,v}z)\bigr]$
  \State $\Psi' \gets \sum_{(u,v)\in\Pi^\star}\bigl[c_{u,v}N_0
         +\sum_{i}c_{u,v}I_{i,v}/(1+c_{u,v}I_{i,v}z)\bigr]$
  \State $z \gets z+(\Lambda-\Psi)/\Psi'$
\Until{convergence}
\State $\mathcal{T}_{\text{H}}^{\text{exact}}(\Pi^\star) \gets 1/z$
\State \textbf{Update} $P_u \gets
       \bar{\Phi}_u \mathcal{T}_{\text{H}}^{\text{exact}}(\Pi^\star)/\tau$ for all
       $u\in\Pi^\star$
\State \Return $\mathcal{T}_{\text{H}}^{\text{exact}}(\Pi^\star)$, $\{P_u\}$
\end{algorithmic}
\end{algorithm}

For HOPE the refined bottleneck latency is
$\mathcal{T}_{\text{H}}^{\text{exact}}(\Pi^\star)$ itself; for THOPE the refined
total latency is
$\mathcal{T}_{\text{H}}^{\text{exact}}(\Pi^\star)+|\Pi^\star|\tau$, since the
cumulative transmission time is unaffected. Each iteration costs
$O(|\Pi^\star|\,|\mc{I}_0|)$ operations, and by Proposition~\ref{prop:concave}
the procedure may be stopped at any point; in our experiments at most four
iterations were required for a relative tolerance of $10^{-12}$, and a single
iteration already recovered essentially the entire gap.

\section{Numerical Results}
\label{sec:results}
To validate the proposed optimal routing framework HOPE and THOPE, we simulate a wireless network with $|\mathcal{N}|$ system nodes and $|\mathcal{I}_0|$ interferers uniformly distributed in a $d \times d$ square area. The source node $S$ is placed at $(0, 0)$ and the destination $D$ at $(d, d)$. 

Unless stated otherwise, we use the following default parameters:  $d = 10~\mathrm{m}$, path-loss exponent $\alpha = 2.2$, noise power is $N_0 = 10^{-6} \mathrm{W}$, interferer transmit power $P_i = 0.1 \mathrm{W}$, transmission duration per hop $\tau = 0.1~\mathrm{s}$, SINR threshold $\gamma = -10~\mathrm{dB}$, reference distance $d_0 = 1~\mathrm{m}$, communication range $r_c = d$, energy conversion efficiency $\eta = 0.8$ in the linear harvesting model $\bar{\Phi}_u=\eta I_u$ of Section~\ref{sec:system}, and target end‑to‑end outage probability $\pie = 0.2$. All results are averaged over $1000$ independent random network realizations. These values are normalized and illustrative: the path-loss model carries no reference gain, which rescales the harvesting term relative to $\tau$, so the results show how the schemes compare as the ratio of harvesting to transmission cost varies, rather than absolute latencies for a specific deployment.

We first illustrate the trade-off between harvesting and transmission delay on
a network instance in which the two objectives select substantially
different routes (Section~\ref{sec:snapshots}).

\subsection{Network Snapshots}
\label{sec:snapshots}

\begin{figure}[t]
\centering
\includegraphics[width=0.7\columnwidth]{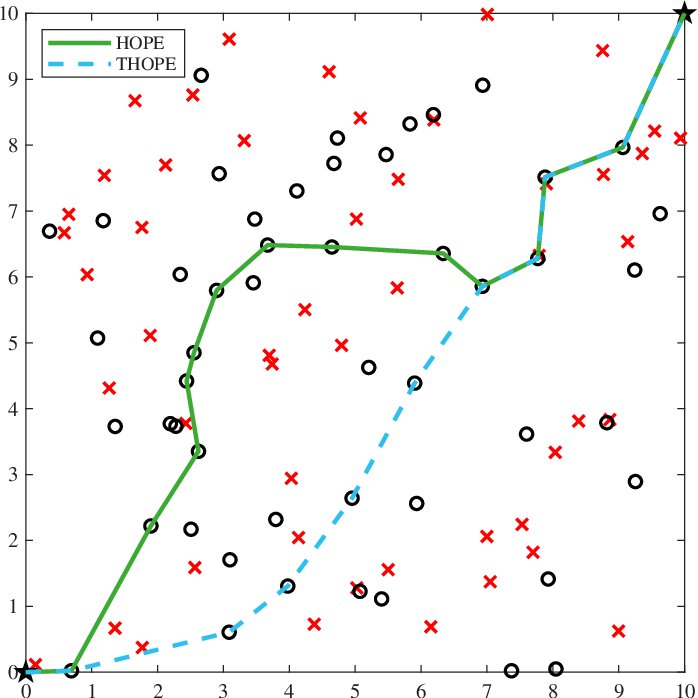}
\caption{Comparison of selected paths in an instance where the hop-count penalty changes the selected route substantially: 
$50$ system nodes (circles) and $50$ interferers ($\times$) uniformly distributed. The green solid line shows the HOPE path (14 hops, total latency $3.11$~s) and the blue dashed line the THOPE path (9 hops, total latency $2.82$~s). HOPE minimizes the harvesting bottleneck and may select longer paths with better energy‑harvesting opportunities; THOPE trades a slight increase in harvesting time for significantly fewer transmission hops, lowering total end‑to‑end latency.
}
\label{fig:network_snapshot}
\end{figure}

Fig.~\ref{fig:network_snapshot} shows an instance in which the two objectives select different routes with substantially different hop counts. In this example, THOPE achieves a total latency of $2.82$~s, while HOPE yields $3.11$~s, but with a different path. HOPE selects a $14$‑hop route because its objective focuses solely on the harvesting bottleneck: additional hops that lie in regions of strong ambient interference can shorten the maximum harvesting time even if they increase the total number of transmissions. THOPE, on the other hand, minimizes the sum of the harvesting bottleneck and the cumulative transmission delay. It therefore prefers the $9$‑hop path, where the reduction in transmission time compensates for a modest increase in harvesting time. This example highlights the distinct trade‑offs between the two design criteria. 

\subsection{Benchmark Strategies}
\label{sec:benchmark}

To quantify the benefits of joint routing and power optimization, we compare HOPE and THOPE against two natural benchmarks:

\textbf{Shortest Path (SP) routing with uniform outage allocation:}  
    This scheme selects the path $\Pi_{\text{SP}}$ that minimizes the Euclidean distance, i.e., edge weights are $d_{u,v}$.  
    To meet the end‑to‑end outage constraint $\pie$, we divide the outage budget equally among the $L$ hops of $\Pi_{\text{SP}}$, setting the per‑link outage target to $\pi_L = 1-(1-\pie)^{1/L}$.  
    Applying the same upper bound as in Lemma~\ref{lem:e2e_outage}, the minimum required transmit power for node $u$ on this path is
    \begin{equation}
        P^{\text{SP}}_u = \frac{\gamma\, d_{u,v}^\alpha (N_0 + I_{v})}{-\ln(1-\pi_L)}.
    \end{equation}
    The corresponding harvesting latency is the maximum of the individual $T_u = P^{\text{SP}}_u\tau/(\bar{\Phi}_u)$:
    \begin{equation}
        \mathcal{T}_H(\Pi_{\text{SP}}) = \max_{u\in \Pi_{\text{SP}}} \frac{P^{\text{SP}}_u\tau}{\bar{\Phi}_u}.
        \label{eq:sp_latency_1}
    \end{equation}
    The total latency adds the transmission time over $L_{\text{SP}}$ hops, i.e. $\mathcal{T}_{\text{tot}}(\Pi_{\text{SP}}) = \mathcal{T}_H(\Pi_{\text{SP}}) + L_{\text{SP}}\,\tau$.

\textbf{Minimum Energy Routing (MER) with outage constraints}: This benchmark minimizes the total transmit energy subject to the additive outage constraint~\cite{sheikholeslami2016energy}.  
    Each link $(u,v)$ is assigned a cost $\sqrt{d_{u,v}^\alpha(N_0+I_v)}$, and the transmit power for node $u$ on the optimal MER path $\Pi_{\text{MER}}$ is
    \begin{equation*}
        P^{\text{MER}}_u = \frac{1}{\epsilon} \sqrt{d_{u,v}^\alpha (N_0 + I_v)}
        \!\!\sum_{(u',v')\in\Pi_{\text{MER}}}\!\!
          \sqrt{d_{u',v'}^\alpha (N_0 + I_{v'})}.
    \end{equation*}
    MER does not explicitly account for harvesting asymmetry, but by reducing transmit energy it indirectly reduces the required harvested energy.  
    The harvesting latency is
    \begin{equation}
        \mathcal{T}_H(\Pi_{\text{MER}}) = \max_{u\in \Pi_{\text{MER}}} \frac{P^{\text{MER}}_u\tau}{\bar{\Phi}_u},
        \label{eq:mer_latency_1}
    \end{equation}
    and the total latency with $L_{\text{MER}}$ hops is
    $\mathcal{T}_{\text{tot}}(\Pi_{\text{MER}}) = \mathcal{T}_H(\Pi_{\text{MER}}) + L_{\text{MER}}\,\tau$.

\subsection{Performance Comparison to Benchmarks}
\label{sec:benchmark_compare}

We now compare the average total latency of the four algorithms over $1000$
network realizations. The ordering is the same at every operating point we
tested: THOPE is fastest, followed by HOPE, with MER next and SP slowest. 

Only THOPE and HOPE price harvesting and interference in the link weight, so
only they can place transmitters where the ambient field is strong; MER, by
minimizing transmit energy, lightens the harvesting burden indirectly and so
trails them while still beating SP, which ignores interference altogether. The
two sweeps below therefore report how each parameter moves the curves rather
than re-establishing that ranking.

Fig.~\ref{fig:number_interferers_all} varies the number of interferers from
$5$ to $100$ with $50$ system nodes. The delay of HOPE and THOPE is low and
nearly flat, because two opposing effects cancel: a higher interferer density
raises the aggregate interference $I_v$ at the receivers, which degrades the
SINR and hence the transmit power required, and at the same time raises $I_u$
at the transmitters, which shortens the time needed to accumulate that energy.
Beyond a moderate density the two approximately balance. SP, which exploits
neither effect, is the most sensitive to interferer density.

\begin{figure}[t]
\centering
\includegraphics[width=0.75\columnwidth]{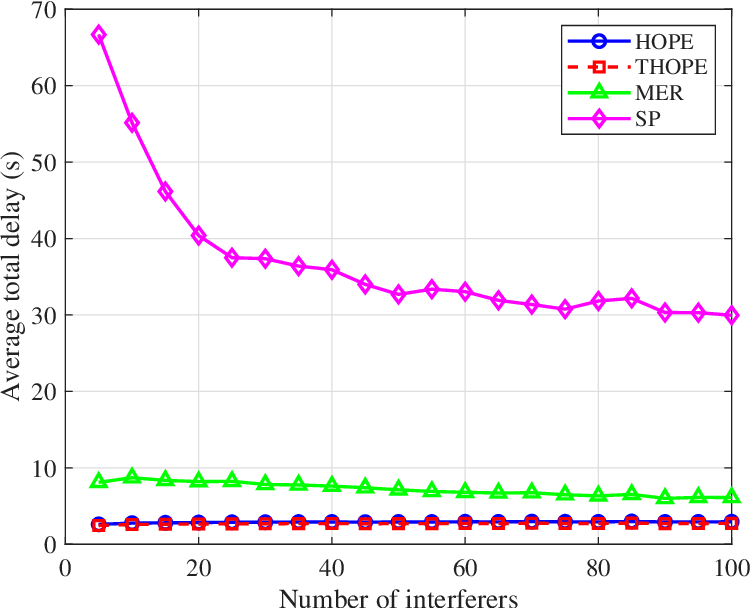}
\caption{Average total delay versus number of interferers for HOPE, THOPE, MER, and SP. The number of system nodes is $50$.}
\label{fig:number_interferers_all}
\end{figure}

Fig.~\ref{fig:number_relays_all} varies the number of system nodes from $25$
to $100$ with $50$ interferers. The delays of HOPE, THOPE and MER all decrease
as more relays become available, since a denser deployment offers a richer set
of candidate paths, from which a route with better harvesting conditions and
fewer hops can be assembled. SP's delay is almost constant, as its route is
determined by geometry alone.

\begin{figure}[t]
\centering
\includegraphics[width=0.75\columnwidth]{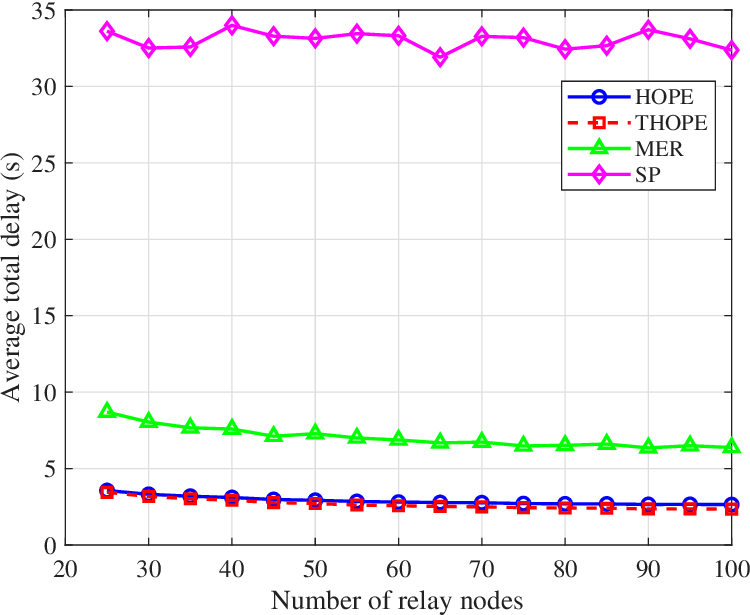}
\caption{Average total delay versus number of system nodes for HOPE, THOPE, MER, and SP. The number of interferers is $50$.}
\label{fig:number_relays_all}
\end{figure}

\subsection{Total Latency Histograms}
\label{sec:histograms}
To provide insight into the distribution of the delay experienced across different network realizations, Fig.~\ref{fig:histogram} plots histograms of the total latency for each algorithm over $1000$ realizations (default parameters). 

The median total latency values are $2.65$~s (THOPE), $2.86$~s (HOPE), $6.52$~s (MER), and $27.46$~s (SP). 
The distributions for HOPE and THOPE are sharply concentrated around their medians, indicating highly predictable performance. 
In contrast, SP and, to a lesser extent, MER exhibit heavy tails. This predictability makes HOPE and THOPE particularly attractive for delay‑sensitive applications with strict latency guarantees.

\begin{figure}[t]
\centering
\includegraphics[
clip=true,
trim={1.5cm 1.2cm 2cm 1.2cm},
width=0.8\columnwidth]{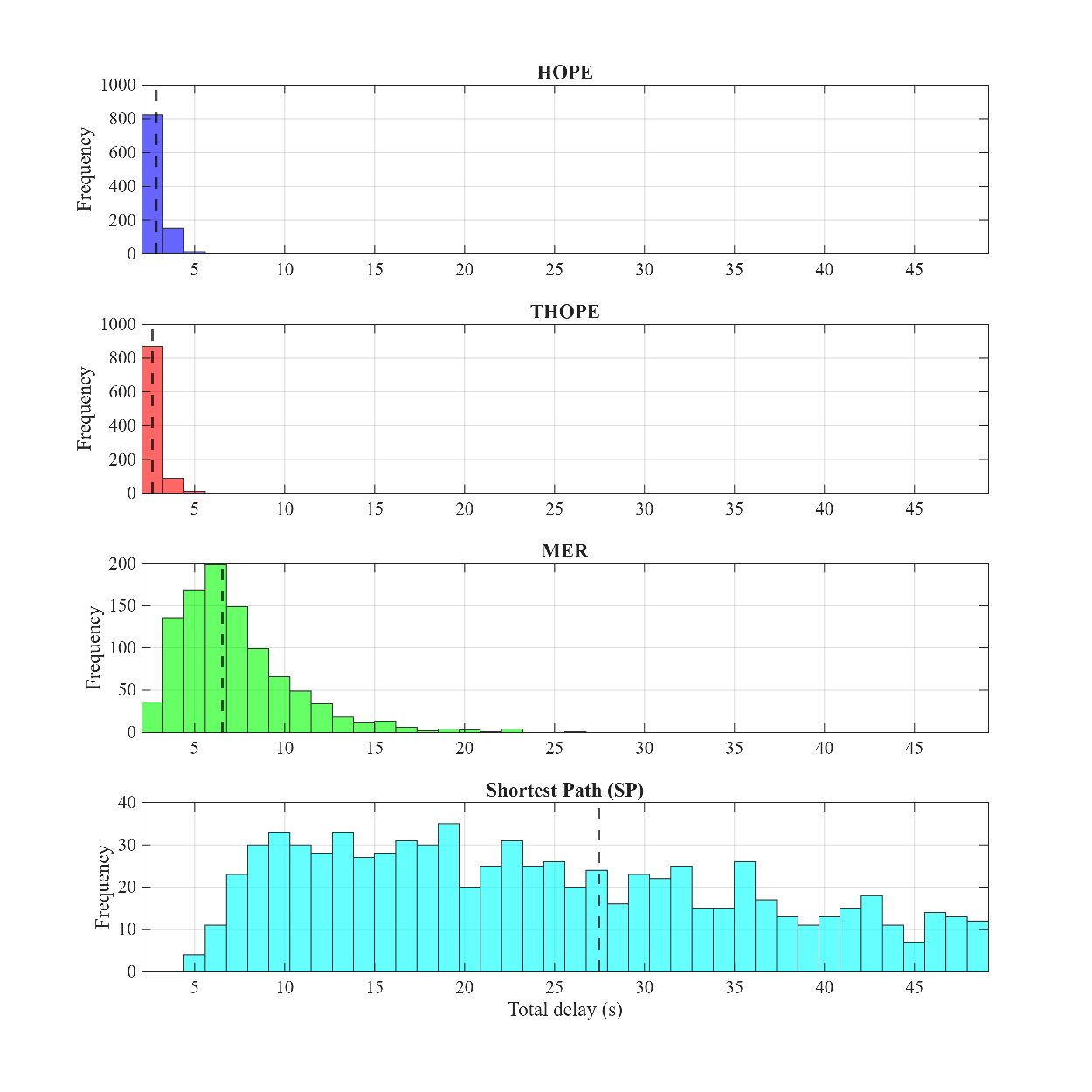}
\caption{Histograms of total latency over $1000$ network realizations for (a) HOPE, (b) THOPE, (c) MER, and (d) SP. 
}
\label{fig:histogram}
\end{figure}

\subsection{Tightness of the Additive Outage Constraint}
\label{sec:tightness_numerics}

We now quantify empirically the latency gap analyzed in
Section~\ref{sec:tightness}. For each network realization we run HOPE and
THOPE, apply the refinement of Algorithm~\ref{alg:refine} to the selected path to find $\mathcal{T}_{\text{H}}^{\text{exact}}$,
and record the relative latency reduction
$1-\mathcal{T}_{\text{H}}^{\text{exact}}/\mathcal{T}_{\text{H}}$.

Table~\ref{tab:tightness} reports the mean reduction of the bottleneck
harvesting latency as a function of $\pie$, together with the analytical cap
of
Theorem~\ref{thm:gap} and the per-instance bound
of Proposition~\ref{prop:instbound}, evaluated on the route each algorithm
returns. 
The results are averaged over $1000$ network realizations with the default parameters. Since
Proposition~\ref{prop:instbound} is a per-path statement, the bound is
evaluated separately on the route each algorithm returns. It held in every
realization, and its hypothesis $G(\Pi)\le1/2$ was satisfied throughout (the largest value of $G(\Pi)$ observed was $0.102$, at $\pie=0.5$)

Multiple conclusions stand out from Table~\ref{tab:tightness}. First, from the measured (exact) latency column, it is apparent that the additive constraint of Lemma~\ref{lem:e2e_outage} is very tight for all values of $\pie$. For instance, at  $\pie=0.2$ the mean recoverable latency is $0.19\%$.
Also, the spread about that mean is narrow with a $95$th percentile of $0.40\%$.

Second, although both Theorem~\ref{thm:gap} and Proposition~\ref{prop:instbound} give bounds on the reduction in latency, Proposition~\ref{prop:instbound} gives much tighter bounds.
In most cases, these bounds are identical to the average measured reduction, and are about two orders of magnitude below the pessimistic  caps  given by Theorem~\ref{thm:gap}. 

\begin{table}[t]
\centering
\caption{}
\label{tab:tightness}
\begin{tabular}{c c c c c c}
\toprule
& & \multicolumn{2}{c}{HOPE route} & \multicolumn{2}{c}{THOPE route}\\
\cmidrule(lr){3-4}\cmidrule(lr){5-6}
$\pie$ & Thm.~\ref{thm:gap} & Prop.~\ref{prop:instbound} & measured
       & Prop.~\ref{prop:instbound} & measured\\
\midrule
$0.05$ & $2.63\%$  & $0.04\%$ & $0.04\%$ & $0.05\%$ & $0.05\%$\\
$0.10$ & $5.56\%$  & $0.09\%$ & $0.09\%$ & $0.10\%$ & $0.10\%$\\
$0.20$ & $12.50\%$ & $0.19\%$ & $0.19\%$ & $0.25\%$ & $0.25\%$\\
$0.30$ & $21.43\%$ & $0.31\%$ & $0.31\%$ & $0.47\%$ & $0.46\%$\\
$0.50$ & $50.00\%$ & $0.61\%$ & $0.59\%$ & $1.20\%$ & $1.14\%$\\
\bottomrule
\end{tabular}
\end{table}

Finally, the reduction in latency grows approximately linearly in $\pie$, mirroring the cap, but
remains roughly two orders of magnitude below it throughout. 

\section{Conclusion and Future Work}
\label{sec:conclusion}

We have studied joint routing and power allocation for multi-hop networks powered entirely by ambient radio-frequency interference. Replacing the
multiplicative end-to-end outage constraint by an additive sufficient condition makes the resulting min--max problem exactly solvable: HOPE
minimizes the bottleneck harvesting latency and THOPE the total end-to-end latency, each through a single shortest-path computation with interference-aware link weights and closed-form transmit powers. 

Because the construction rests on a relaxation, we also quantified its cost. 
In Theorem~\ref{thm:gap}, we proved that the achievable latency reduction is capped by $\pie/\bigl(2(1-\pie)\bigr)$ on every path.
We also proved that a tight per-instance bound is
computable from the algorithms' own output and the geometry of the interferers' locations (Proposition~\ref{prop:instbound}). 

Numerically, we showed that HOPE and THOPE reduce average total latency by up to an order of
magnitude relative to both benchmarks, and their delay distributions are sharply concentrated around the median. This matters for applications with hard deadlines.
Also, we showed that at the outage target of $\pie=0.2$,  the measured reduction in HOPE's bottleneck latency with respect to the latency when the exact outage formula is used is $0.19\%$. The average per-instance bound from  Proposition~\ref{prop:instbound} is also $0.19\%$, while the analytical reduction cap of Theorem~\ref{thm:gap} is $12.5\%$.

Three extensions suggest themselves. The decoding threshold $\gamma$ is treated as fixed, but it can be optimized. The energy-neutrality
constraint in this paper reflects a worst case in which no
energy is carried across intervals; permitting accumulation in the supercapacitor buffer would shift the metric from first-packet latency to
steady-state cycle time.
Finally, the single-flow assumption could be
relaxed by scheduling multiple concurrent flows, whose transmissions would then interfere with and could be harvested from one another.

\bibliographystyle{IEEEtran}  
\bibliography{refs} 

\end{document}